\documentclass[twocolumn]{autart}    
\let\theoremstyle\relax

\usepackage{amsthm}
\usepackage{graphicx}          
\usepackage{amsmath}
\usepackage{amssymb}
\usepackage{bbm}
\usepackage{mathrsfs}

\newcommand{\R}{\mathbb{R}} 
 
\newcommand{\dd}   {{\rm d}\hbox{\hskip 0.5pt}}

\newcommand{\sbm}[1]{\left[\begin{smallmatrix} #1
   \end{smallmatrix}\right]}

\newcommand{\rfb}[1]{\mbox{\rm
   (\ref{#1})}\ifx\undefined\stillediting\else:\fbox{$#1$}\fi}

\newtheorem{theorem}           {Theorem}[section] 
\newtheorem{definition}         [theorem]{Definition}

\newtheorem{proposition}		[theorem]{Proposition} 
\newtheorem{corollary}		[theorem]{Corollary}
 
\newtheorem{example}		[theorem]{Example}
\newtheorem{remark}	      [theorem]{Remark}

\begin{document}

\begin{frontmatter}

\title{Output contraction analysis of nonlinear systems} 


\author[Paestum]{Hao Yin}\ead{yhyinhao02122575@163.com},    
\author[Paestum]{Bayu Jayawardhana}\ead{b.jayawardhana@rug.nl},               
\author[Baiae]{Stephan Trenn}\ead{s.trenn@rug.nl}  

\address[Paestum]{Discrete Technology $\&$ Production Automation, University of Groningen}  
\address[Baiae]{Systems, Control and Optimization; Bernoulli Institute, University of Groningen}        

\begin{abstract}                          
This paper introduce the notion of output contraction that expands the contraction notion to the time-varying nonlinear systems with output. It pertains to the systems' property that any pair of outputs from the system converge to each other exponentially. {This concept exhibits a more expansive nature when contrasted with another generalized contraction framework known as partial contraction.} The first result establishes a connection between the output contraction of a time-varying system and the output exponential stability of its variational system. Subsequently, we derive a sufficient condition for achieving output contraction in time-varying systems by applying the output contraction Lyapunov criterion. Finally, we apply the results to analyze the output exponential stability of nonlinear time-invariant systems.
\end{abstract}

\end{frontmatter}

\section{Introduction}
The notion of contraction, also referred to as exponential incremental stability, 
provides an analytical tool to study
the asymptotic behavior of time-varying nonlinear systems. The contraction theory characterizes particular systems' property of nonlinear time-varying systems where 
all trajectories from different initial conditions converge exponentially to each other 
\cite{lohmiller1998contraction} and it has been studied extensively in literature for over two decades. 
Forni and Sepulchre \cite{forni2013differential} 
present a Lyapunov-like characterization of contraction property by lifting the Lyapunov function to the tangent bundle. A closely related notion, the so-called incremental stability, is studied in 
\cite{angeli2002lyapunov} via incremental Lyapunov functions. 
As an alternative to the Lyapunov approach, Sontag et al. \cite{sontag2014three} employ matrix measure, also known as logarithmic norm, to directly characterize the contraction properties. In \cite{davydov2022non}, the authors present weak pairings and explore the one-sided Lipschitz condition for the vector field to investigate contractivity with respect to arbitrary norms. The authors in \cite{andrieu2016transverse} study transverse exponential stability, which can be considered as a generalized concept of contraction by using nonlinear Rieammanian metrics.

In general, the contraction notion 
refers to the contraction property of the entire state variables 
\cite{lohmiller1998contraction,forni2013differential,angeli2002lyapunov,sontag2014three,davydov2022non,andrieu2016transverse}. However, 
establishing the contraction property of only a portion of the state variables 
is not trivial. 
Wang et al. in \cite{wang2005partial} introduces the concept of partial contraction, 
which studies the contractivity of subsets of the state space. In \cite{forni2013differential}, the concept of horizontal contraction is presented to study contraction property along particular directions. The notion of $k$-contraction, as defined in \cite{wu2022k}, establishes a contraction property applicable to a $k$-parallelotope. A recent study in \cite{wu2023partial} shows that under certain mild assumptions, partial contraction can lead to horizontal contraction, and horizontal contraction subsequently implies $k$-contraction. In this paper, we explore a general class of systems that admit non-contractive systems as commonly found in practice but can exhibit contraction property for particular output mappings. 

In control theory, the role of systems' output is important as 
it represents the observable or measurable state of the system, with which we determine the control law to steer the input variable. 
Accordingly, many control problems are defined based on the use of systems' output, 
such as output regulation problems \cite{isidori1990output} and output tracking control problems \cite{gao2008network}. In this paper, we 
introduce the notion of output contraction to nonlinear time-varying systems with output, 
which posits that any pair of systems' output converge to each other exponentially. 
This notion generalizes the classical notion of contraction to the class of systems that may not exhibit contraction properties; we will illustrate this later in Example \ref{ex1}. {Additionally, we use a simple contraction case to illustrate that output contraction showcases a broader class of systems compared to the notion of partial contraction introduced in \cite{wang2005partial}.} In our first main result, we establish that the property of 
{Output Exponential Stability (OES)} (which will be defined later) of the variational dynamics is a necessary and sufficient condition for the output contraction of the original systems. If the output map corresponds to an identity map, this condition can be simplified to match the results presented in \cite[Prop.~1]{barabanov2019contraction}, reducing output contraction to standard contraction. In our second result, we present sufficient conditions for output contraction of nonlinear time-varying systems by utilizing an output contraction Lyapunov function to analyze the OES property of the variational systems. In our third contribution, 
we employ the output exponential stability Lyapunov condition to assess the OES of time-invariant systems. This is achieved by ensuring that the system exhibits output contraction, with a constant output 
residing within its output set.

The paper is organized as follows. In Section 2, we present
preliminaries and the main problem formulation. Our main results are presented 
in Section 3, where we present necessary and sufficient conditions
for the output contraction of nonlinear time-varying systems, and the output contraction Lyapunov criterion. The numerical simulations and applications are provided in Section 4 and the conclusions are given in Section 5.

\section{Preliminaries and problem formulation}
Throughout this paper, we consider the following nonlinear time-varying systems 
\begin{equation}\label{odet}
 \left\{\begin{matrix}
\dot{x}=f(x,t),\\ 
y=h(x,t),
\end{matrix}\right.   
\end{equation}
where $x(t)\in \mathbb{R}^{n}$ is the state vector, $f:\mathbb{R}^n\times \mathbb{R}_+ \rightarrow \mathbb{R}^n$ is the vector field, and $h:\mathbb{R}^m\times \mathbb{R}_+ \rightarrow \mathbb{R}^m$ describes the output map. We assume that $f$ and $h$ are continuously differentiable. 
\begin{definition}\label{d1}
A time-varying system \eqref{odet} is {called {\em output contractive}} with respect to the state if there exists positive numbers $c$ and $\alpha$ such that for any pair of the outputs $y(t),y'(t)\in \mathbb{R}^m$ of \eqref{odet} with $i=1,2$, we have
\begin{equation} \label{eq:d1}
\begin{aligned}
\|
y(t)-y'(t)
\|\leq ce^{-\alpha (t-t_0)}\|
x(t_0)-x'(t_0)
\|, \quad \forall t\geq t_0.
\end{aligned}
\end{equation}
\end{definition}
Note that, if the output $h(x,t)$ is an identity map, i.e. $y=x$, Definition \ref{d1} reduces to the standard contraction notion as in \cite[Def.~1]{barabanov2019contraction}. 
\begin{remark}\label{remark1}
A concept closely related to the {\em output contraction} in Definition \ref{d1}, is the {\em partial contraction}, as introduced in \cite{wang2005partial,wu2023partial}. In these works,  
equation \eqref{eq:d1} is replaced by \cite[Def.~4]{wu2023partial}
\begin{equation} \label{eq:d10}
\|y(t)-y'(t)
\|\leq ce^{-\alpha (t-t_0)}\|
y'(t_{0})-y'(t_{0})\|,
\end{equation} 
where exponential decay is only scaled by the initial output difference and not by the entire initial state difference. The following example shows that a system can be an {\em output contraction} while it is not a {\em partial contraction}. 
Consider the following stable LTI system 
\begin{equation} \label{ex00}
\begin{aligned}
\left\{\begin{matrix}
\dot{x}_{1}=-2x_1+x_2,\\ 
\dot{x}_{2}=x_1-2x_2.
\end{matrix}\right.
\end{aligned}
\end{equation}
The trajectories of the system are given by $x_1(t)=\frac{1}{2}(x_{10}-x_{20})e^{-3t}+\frac{1}{2}(x_{10}+x_{20})e^{-t}$ and $x_2(t)=\frac{1}{2}(x_{10}-x_{20})e^{-3t}x_{20}+\frac{1}{2}(x_{10}+x_{20})e^{-t}$. If we take $y(t) = x_1(t)$ then $\|y(t)-y'(t)\|=\|-\frac{1}{2}(y(0)-y'(0)-x_{20}+x'_{20})e^{-3t}+\frac{1}{2}(y(0)-y'(0)+x_{20}-x'_{20})e^{-t}\|$. Consequently, we can easily find $c>0$ and $\alpha>0$ such that \eqref{eq:d1} holds, but because the output difference $y-y'$ depends on the difference of the initial condition of $x_2$ and $x'_2$ it is not possible to satisfy \eqref{eq:d10}. 
\end{remark}
Observe that output contraction  depends heavily on the time-varying nature of the output mapping $h(x, t)$. In fact, it is neither necessary nor sufficient for output contraction that the systems' state is contractive. For instance, consider again the aforementioned example provided in \eqref{ex00}. If we take the output map as $y = e^{2t}x_1$, it is clear that the system no longer possesses the property of output contraction. For a converse example (i.e., output contraction but not state contraction), we refer to the forthcoming Example~\ref{ex1}.

Similar to the traditional analysis of contraction, we investigate the output contraction of \eqref{odet} by analyzing the output stability of the associated variational system with outputs. This variational system, connected to the systems in \eqref{odet}, is given as follows
\begin{equation}\label{vari}
 \left\{\begin{matrix}
\dot{\xi }=\frac{\partial f}{\partial x}(x(t),t)\cdot \xi,\\ 
\nu=\frac{\partial h}{\partial x}(x(t),t)\cdot \xi,
\end{matrix}\right.   
\end{equation}
where $x(t)$ is any solutions of \eqref{odet}. We omit the explicit dependence on $(x(t),t)$ whenever it is clear from the context. Note that \eqref{vari} actually denotes a whole family of time-varying linear systems which is parameterized by the initial value $x(t_0)$ of \eqref{odet}.  

\begin{definition}\label{d2}
The variational system \eqref{vari} is called {\em 
{Output Exponentially Stable (OES)} with respect to the state}, if there exist positive numbers $c$, $\alpha$ {(independent of $x$, $t_0$ and $\xi(t_0)$)} such that for every output $\nu(t)\in\mathbb{R}^m$ of \eqref{vari} the inequality
\begin{equation} \label{eq:d3}
\begin{aligned}
\|\nu(t)\| \leq ce^{-\alpha (t-t_0)}\|\xi(t_0)\|, 
\end{aligned}
\end{equation}
holds for all $t\geq t_0$. 
\end{definition}

\section{Main result}\label{sec:main_result}
In this section, we firstly establish an equivalent relationship between the output contraction of a time-varying system \eqref{odet} and the OES of its variational system \eqref{vari}. Secondly, a sufficient condition is presented that guarantees the OES of the variational system.
\begin{proposition}\label{proposition1}
The nonlinear time-varying system \eqref{odet} {is {\em output contractive}} with respect to the state if and only if the corresponding variational system \eqref{vari} is {\em OES} with respect to the state.
\end{proposition}
\begin{proof}
Let us first establish a relationship between the solutions of \eqref{odet} and those of \eqref{vari}. Let $\sbm{x(t) \\ y(t)}=\sbm{\varphi(x_0,t) \\ h(\varphi(x_0,t),t)}$ and $\sbm{\hat{x}(t)\\ \hat{y}(t)}=\sbm{\varphi(x_0+\delta \xi _0,t) \\ h(\varphi(x_0+\delta \xi _0,t),t)}$ be two trajectories and outputs of \eqref{odet} with initial conditions $x_0$ and $\hat{x}_0:=x_0+\delta \xi_0$, respectively, where $\delta$ is a sufficiently small positive constant and $\xi_0\neq 0$ ($\hat{x}_0$ and $x_0$ are two different initial conditions) will be related later to the initial condition of \eqref{odet}. In the following, we will show that
\begin{equation}\label{varitaj}
\left\{\begin{matrix}
\xi(t):=\lim\limits_{\delta \rightarrow 0} \frac{\varphi(x_0+\delta \xi _0,t)-\varphi(x_0,t)}{\delta },\\ 
\nu (t):=\lim\limits_{\delta \rightarrow 0} \frac{h(\varphi(x_0+\delta \xi _0,t),t)-h(\varphi(x_0,t),t)}{\delta},
\end{matrix}\right.
\end{equation}
are a solution of \eqref{vari} w.r.t.\ $x(t)$ and with initial value 
\[
\begin{aligned}
\xi(t_0) &= 
\xi_0%
,\\ 
\nu (t_0) &= 
\nu_0=\nu_0(\xi_0) :=
\lim\limits_{\delta \rightarrow 0} \tfrac{h(x_0+\delta \xi_0,t_0)-h(x_0,t_0)}{\delta}.
\end{aligned}
\]
By denoting the partial Jacobian matrices of $\varphi$ at $(x_0,t)$ by $\Phi_{x_0}(t)$, we can then rewrite \eqref{varitaj} as 
\begin{equation}\label{varitaj0}
\left\{\begin{matrix}
\xi(t)=\Phi_{x_0}(t)\cdot \xi_0,\\ 
\nu (t)=\frac{\partial h}{\partial x}\cdot\Phi_{x_0}(t)\cdot\xi_0.
\end{matrix}\right.
\end{equation}
From \eqref{varitaj0}, we know that \eqref{varitaj} satisfies $\nu=\tfrac{\partial h}{\partial x}\xi $ as desired. The flow $\varphi (x_0,t)$ of \eqref{odet} satisfies
\begin{equation}\label{fl1}
\begin{aligned}
&\varphi (x_0,t)=x_0+\int_{t_0}^{t}f\big(\varphi (x_0,\tau), \tau\big)\dd \tau,
\end{aligned}
\end{equation}
and similarly, the flow $\varphi(x_0+\delta\xi_0,t)$ satisfies
\begin{equation}\label{fl2}
\begin{aligned}
&\varphi(x_0+\delta\xi_0,t)=x_0+\delta\xi_0+\int_{t_0}^{t}f\big(\varphi (x_0+\delta\xi_0,\tau),\tau\big)\dd \tau.
\end{aligned}
\end{equation}
Hence,
\begin{equation}\label{varitaj1}
\begin{aligned}
&\xi (t)=\xi_0+\int_{t_0}^{t}\lim_{\delta \rightarrow 0}\frac{1}{\delta }\Big(f\big(\varphi (x_0+\delta\xi_0,\tau),\tau\big)\big)\\& \qquad -f\big(\varphi (x_0,\tau), \tau\big)\Big)\dd \tau
\end{aligned}
\end{equation}
Clearly,
\begin{equation}\label{varitaj2}
\begin{aligned}
&\lim_{\delta \rightarrow 0}\frac{1}{\delta }\Big(f\big(\varphi (x_0+\delta\xi_0,\tau),\tau\big)\big)-f\big(\varphi (x_0,\tau), \tau\big)\Big)
\\&=\tfrac{\partial f}{\partial x}(x(\tau),\tau)\cdot \Phi_{x_0}(\tau)\cdot \xi_0
\\&\overset{\eqref{varitaj0}}{=} \tfrac{\partial f}{\partial x}(x(\tau),\tau)\cdot \xi(\tau).
\end{aligned}
\end{equation}
Substituting this back to \eqref{varitaj1} and differentiating with respect to time gives us
\begin{equation}\label{varitaj8}
\begin{aligned}
\dot{\xi }(t)=\frac{\partial f}{\partial x}\xi(t).
\end{aligned}
\end{equation}
Altogether this shows that indeed $\xi(t)$, $\nu(t)$ given by \eqref{varitaj} is a solution of \eqref{vari}. We can now show the sufficiency result.

\emph{Output Contraction $\Rightarrow$ OES.} Let $c$ and $\alpha$ be the constants corresponding to the output contractivity condition. Seeking a contradiction, assume the variational system \eqref{vari} is not OES. Then there exists a solution $x(\cdot)$ of \eqref{odet} and an initial value $\xi_0$ such that for the corresponding output $\nu(\cdot)$ of \eqref{vari} we have that for $c':=\frac{3}{2}c$ and $\alpha':=\alpha$, there exists $T>0$ such that 
\begin{equation} \label{pf90}
\left\|\nu(T)\right\| > c'e^{-\alpha' (T-t_0)} \left\|\xi_0\right\| = \frac{3}{2}ce^{-\alpha (T-t_0)} \left\|\xi_0 \right\|.
\end{equation} 
Let $\hat{y}(\cdot)$, $y(\cdot)$ be outputs of \eqref{odet} with initial values $\hat{x}(t_0):=x(t_0)+ \delta \xi(t_0)$ and $x(t_0)$, respectively. By definition, $\nu (t):=\lim\limits_{\delta \rightarrow 0} \frac{\hat{y}(t)-y(t)}{\delta}$.
Hence, for a sufficiently small $\delta>0$, we have that at time $T$, 
\begin{equation} \label{pf72}
\begin{aligned}
 \frac{\left\|
\hat{y}(T)-y(T)
\right\|}{\delta }>\frac{2}{3}\left\|\nu(T)\right\|,
\end{aligned}
\end{equation}
where the lower-bound constant $\frac{2}{3} < 1$ is chosen arbitrarily for the following computation of bounds. Combining \eqref{pf90}, \eqref{pf72}, we obtain
\begin{align} 
\left\|
\hat{y}(T)-y(T)
\right\| &\overset{\eqref{pf72}}{>}\frac{2}{3}\delta\left\|\nu(T)\right\|\overset{\eqref{pf90}}{>}c e^{-\alpha (T-t_0)}\left\| 
\delta \xi_0\right\|\nonumber\\&=ce^{-\alpha (T-t_0)}\left\|
\hat{x}(t_0)-x(t_0) 
\right\|
\nonumber,
\end{align}
for all $\xi_0\neq 0$. This is in contradiction to the output contractivity of \eqref{odet} and concludes the proof of the sufficiency part.

\emph{OES $\Rightarrow$ Output Contraction.} Let us consider two outputs $y(\cdot)= h\big(\varphi(x_0,\cdot)\big)$ and $\hat{y}(\cdot)= h\big(\varphi(\hat{x}_0,\cdot)\big)$ of \eqref{odet}. Consequently, we can utilize the fundamental theorem of calculus for line integrals to obtain
\begin{equation}\label{pf13}
\hat{y}(t) - y(t)=\int_{x_0}^{\hat{x}_0} \frac{\dd h\big(\varphi(\zeta,t)\big)}{\dd \zeta}\dd \zeta 
\end{equation}
According to \eqref{varitaj}, one has
\begin{equation}\label{pf131}
\nu(t)=\frac{\dd h\big(\varphi(x_0,t)\big)}{\dd x_0}\xi_0.
\end{equation}
From the OES property of \eqref{vari} and \eqref{pf131}, one has
\begin{equation} \label{pf16}
\begin{aligned}
\left\|\nu(t)\right\|=\left\|\frac{\dd h\big(\varphi(x_0,t)\big)}{\dd x_0}\xi_0\right\|\leq ce^{-\alpha (t-t_0)}\|\xi_0\|.
\end{aligned}
\end{equation}
Since $\xi_0\neq 0$, we have
\begin{equation} \label{pf17}
\begin{aligned}
\frac{\left\|\frac{\dd h\big(\varphi(x_0,t)\big)}{\dd x_0}\xi_0\right\|}{\|\xi_0\|}\leq ce^{-\alpha (t-t_0)}.
\end{aligned}
\end{equation}
Given that $\xi_0\neq 0$ is chosen arbitrarily, it follows from \eqref{pf17} that
\begin{equation} \label{pf18}
\begin{aligned}
\left\|\frac{\dd h\big(\varphi(x_0,t)\big)}{\dd x_0}\right\|=\underset{\|\xi_0\|\neq 0}{\sup}\frac{\left\|\frac{\dd h\big(\varphi(x_0,t)\big)}{\dd x_0}\xi_0\right\|}{\|\xi_0\|}\leq ce^{-\alpha (t-t_0)},
\end{aligned}
\end{equation}
Using \eqref{pf18} to get the upper bound of \eqref{pf13}, we have
\begin{equation} \label{pf19}
\begin{aligned}
\left\|\hat{y}(t) - y(t)\right\|&=\left\|\int_{x_0}^{\hat{x}_0} \frac{\dd h\big(\varphi(\zeta,t)\big)}{\dd \zeta}\dd \zeta\right\|\\&\overset{\eqref{pf18}}{\leq}c e^{-\alpha (t-t_0)}\left\|
\hat{x}(t_0)-x(t_0) 
\right\|,
\end{aligned}
\end{equation}
This shows that \eqref{odet} is output contracting and the proof is complete.
\end{proof}
We present now the following theorem on the output contraction property of system \eqref{odet}.
\begin{theorem}[Output Contraction Lyapunov Condition]\label{theorem1} Consider the time-varying system \eqref{odet} with its corresponding variational system \eqref{vari}. Suppose that there exist positive constants $\alpha_1, \alpha_2, \alpha_3, \alpha_4\in\mathbb{R}_{\geq 0}$, $\alpha_3< \alpha_4$, $p\in\mathbb{R}_{\geq 1}$ and a continuous function $V:\mathbb{R}^{n}\times \mathbb{R}^{n}\times \mathbb{R}_{\geq 0}\rightarrow \mathbb{R}_{\geq 0}$ such that, for all $(x,\xi,t)\in\R^n\times\R^n\times\R_{\geq 0}$,
\begin{equation} \label{t1}
\begin{aligned}
\alpha_1\|\nu\|^p\leq V(x,\xi,t)\leq \alpha_2\|\xi\|^pe^{\alpha_3 (t-t_0)}
\end{aligned}
\end{equation}
is satisfied, {where $\nu(x,\xi,t)=\frac{\partial h}{\partial x}(x,t)\xi$}, 
and such that
\begin{equation} \label{t2}
\begin{aligned}
\dot{V}:=\frac{\partial V}{\partial t}+\frac{\partial V}{\partial x}f+\frac{\partial V}{\partial \xi}\frac{\partial f}{\partial x}\xi\leq -\alpha_4 V
\end{aligned}
\end{equation}
holds. Then the system \eqref{odet} is \emph{output contraction} according to Definition \ref{d1}.
\end{theorem}
\begin{proof}
Note that for every solution $x$ of \eqref{odet} and corresponding solution $\xi$ of \eqref{vari} we have that $\tfrac{\mathrm{d}}{\mathrm{d} t} V(x(t),\xi(t),t)$ equals the left-hand side of \eqref{t2}. Consequently, we have 
\begin{equation} \label{t3}
\begin{aligned}
V(x(t),\xi(t),t)\overset{\eqref{t2}}{\leq} c'e^{-\alpha_4 (t-t_0)}V(x_0,\xi_0,t_0),
\end{aligned}
\end{equation}
with $c'>0$. Then, it follows from \eqref{t1} that
\begin{equation} \label{t4}
\begin{aligned}
\|\nu\|&\overset{\eqref{t1}}{\leq}(\tfrac{1}{\alpha_1})^\frac{1}{p}V(x,\xi,t)^\frac{1}{p}\\
&\overset{\eqref{t3}}{\leq} (\tfrac{c'}{\alpha_1})^\frac{1}{p}e^{-\frac{\alpha_4}{p}(t-t_0)}V(x_0,\xi_0,t_0)^\frac{1}{p}\\
&\overset{\eqref{t1}}{\leq} (\tfrac{c'\alpha_2}{\alpha_1})^\frac{1}{p}e^{-\frac{\alpha_4-\alpha_3}{p}(t-t_0)}\|\xi_0\|.
\end{aligned}
\end{equation}
The variational system \eqref{vari} is OES with $c=c'\frac{\alpha_1}{\alpha_2}$ and $\alpha=\frac{\alpha_4-\alpha_3}{p}$. By Proposition \ref{proposition1}, we can conclude that the system \eqref{odet} {is output contractive.} 
\end{proof}
In Theorem \ref{theorem1}, our analysis includes the utilization of a time-varying Lyapunov function, which distinguishes \eqref{t1} from the established bounds of Lyapunov function. If we select $\alpha_3=0$, equation \eqref{t1} becomes the standard Lyapunov function bounds. For the LTI system \eqref{ex00} in the preceding example, we can choose $V(x,\xi,t)={\xi_1^2+\xi_2^2}$, with $\alpha_1=\alpha_2=1$, $\alpha_3 = 0$, and $p=\alpha_4=2$. In this case, Theorem \ref{theorem1} can be applied to show that the system \eqref{ex00} {is output contractive}. Observe that the choice of $V(x, \xi, t)$ guarantees the contraction of system \eqref{ex00}. {This example shows that $V(x,\xi,t)$ does not need to depend explicitly on $x$ and in view of the $x$-independent upper bound \eqref{t1}, the possible dependence on $x$ is in fact rather limited in general. However, even if $V$ does not explicitly depends on $x$, $\dot{V}$ will still depend on $x$ via $\tfrac{\partial f}{\partial x}$.}

\section{Simulation setup and applications}
This section presents two numerical examples to numerically validate the proposed analysis tools and demonstrate their application. The first example explores the output contraction property of a non-contracting system using Theorem \ref{theorem1}. The second example showcases the practical utilization of Proposition \ref{proposition1} for assessing the output exponential stability of a time-invariant system, as illustrated by Corollary \ref{corollary1}. 
\subsection{Output contraction of nonlinear time-varying systems}
\begin{example}\label{ex1}
Consider a time-varying system whose dynamics take the form
\begin{equation} \label{ex10}
\left\{
\begin{array}{rl}
\dot{x}_{1}&=-0.1x_1^3-(4+\sin t+0.3x_1^2)x_{2}\\
& \qquad +\sin(x_1+x_2)+\cos t,\\ 
\dot{x}_{2} & =-0.1x_2^3-(4+\sin t++0.3x_2^2)x_{1}\\
& \qquad +\cos(x_1+x_2)+\sin t,\\
y& =x_1+x_2,
\end{array}
\right.
\end{equation}
where $x(t)\in \mathbb{R}^{n}$ is the state vector and $y(t)\in \mathbb{R}^m$ is the output. The left plot of Figure \ref{fig:0} shows the trajectories originating from two different initial conditions $\sbm{-2.5\\-5}$ and $\sbm{-1.5\\-3}$. In this plot, it is clear that 
the system \eqref{ex10} does not exhibit contraction behavior.
\begingroup
\begin{figure}[!htb]
 \centering
 \includegraphics[width=3in]{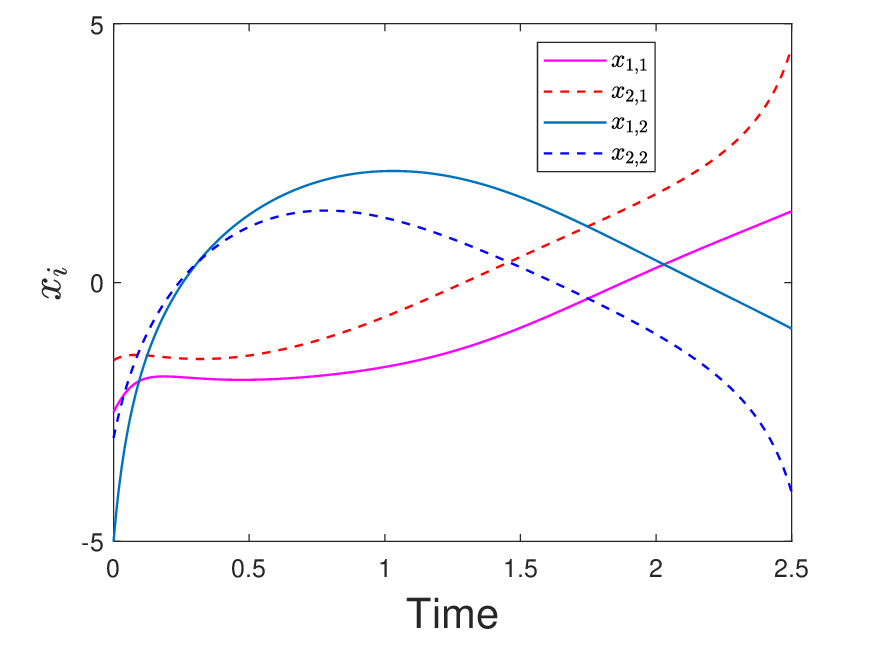}
 \includegraphics[width=3in]{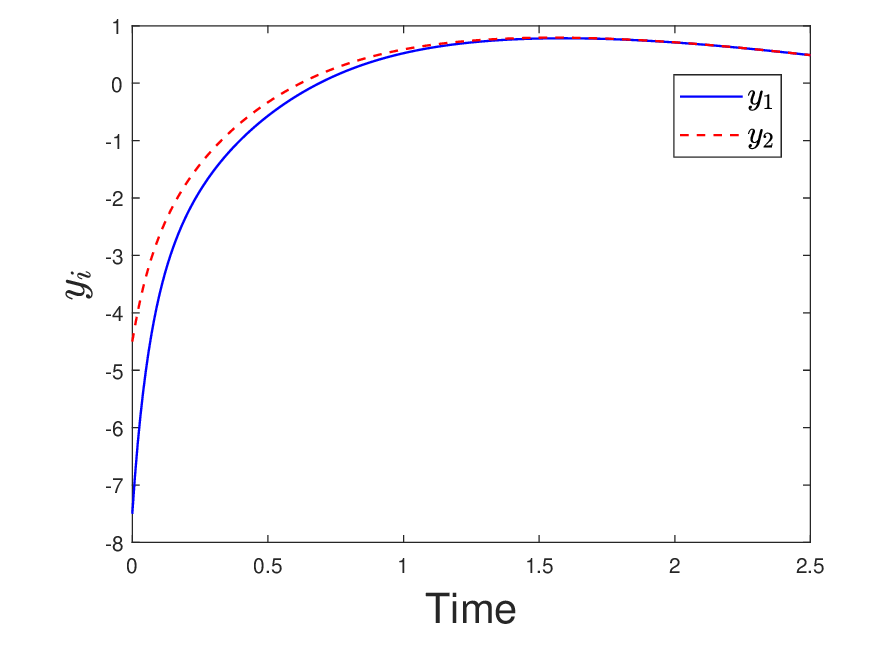}
  \caption{The plot of trajectories and output of time-varying system in Example \ref{ex1}
initialized at $\sbm{-2.5\\-5}$ and $\sbm{-1.5\\-3}$.} \label{fig:0}
\end{figure}
\endgroup
The variational system of \eqref{ex10} is given by
\begin{equation} \label{ex11}
\left\{\begin{array}{rl}
\dot{\xi}_{1} & = -(4+\sin t-0.6x_1x_2)\xi_{2}\\
& \qquad+\big(\cos(x_1+x_2)-0.3x_1^2\big)(\xi_{1}+\xi_{2}),\\ 
\dot{\xi}_{2}& =-(4+\sin t-0.6x_1x_2)\xi_{1}\\ 
& \qquad -\big(\sin(x_1+x_2)-0.3x_2^2\big)(\xi_{1}+\xi_{2}),\\
\nu & =\xi_1+\xi_2.
\end{array}\right.
\end{equation}
Using the output contraction Lyapunov function $V(x,\xi,t) = (\xi_1 + \xi_2)^2$, we can satisfy \eqref{t1} by setting $\alpha_1 = 1$, $\alpha_2 = 2$, $\alpha_3 = 0$, and $p = 2$. The derivative of $V(x,\xi,t)$ satisfies $\dot{V}(x,\xi,t)=-2\big(4-\sin t+\cos(x_1+x_2)+\sin(x_1+x_2)+0.3(x_1+x_2)^2\big)(\xi_1+\xi_2)^2$. Subsequently, \eqref{t2} can be fulfilled by taking $\alpha_4 = 2$. Hence \eqref{ex10} satisfies the hypotheses for the output contraction as outlined in Theorem \ref{theorem1}. The outputs corresponding to two different initial conditions $\sbm{-2.5\\-5}$ and $\sbm{-1.5\\-3}$ are depicted in the right plot of Figure \ref{fig:0}.
\end{example}
\subsection{Output exponential stability of time-invariant systems}
It is well known that for a contracting time-invariant system, all trajectories converge to an equilibrium exponentially. As an interesting particular case of our main results above, we can prove output exponential stability of the time-invariant systems by using Proposition \ref{proposition1}. 
Consider a time-invariant nonlinear autonomous system given by
\begin{equation}\label{ode}
 \left\{\begin{matrix}
\dot{x}=f(x),\\ 
y=h(x),
\end{matrix}\right.   
\end{equation}
where there exists a $x^*(t)$ (not necessarily an equilibrium), {which admits an output equilibrium $y^*$} (possibly unknown), such that $y^*=h(x^*(t))$. 
In the next definition, we extend \cite[Def.~2.3]{karafyllis2021relation} into the nonlinear systems case. 
\begin{definition}\label{d3}
A time-invariant autonomous system \eqref{ode} with an  output equilibrium $y^*$ is called {\em Output Exponentially Stable (OES)} with respect to the state if there exists positive numbers $c$ and $\alpha$ such that for any pair of the output $y(t)\in \mathbb{R}^m$ of \eqref{ode}, 
\begin{equation} \label{eq:d31}
\begin{aligned}
\|y(t)-y^*\|\leq ce^{-\alpha (t-t_0)}\|x(t_{0})-x^*(t_0)\|, \quad \forall t\geq t_0,
\end{aligned}
\end{equation}
where $x(t_0)$, $x^*(t_0)$ are the initial conditions of the state $x(t)$, $x^*(t_0)$, respectively. 
\end{definition}
In the subsequent corollary, we analyze the OES of \eqref{ode} through its variational system \eqref{vario}.
\begin{corollary}[Output Exponential Stability Lyapunov Condition]\label{corollary1}
The time-invariant system \eqref{ode} with its corresponding
variational system \eqref{vario} is {\em OES} if there exist positive constants $\alpha_1, \alpha_2, \alpha_3\in\mathbb{R}_{\geq 0}$, $p\in\mathbb{R}_{\geq 1}$ and a continuous function $V:\mathbb{R}^{n}\times \mathbb{R}^{n}\times\rightarrow \mathbb{R}_{\geq 0}$ such that 
\begin{equation} \label{c1}
\begin{aligned}
\alpha_1\|\nu\|^p\leq V(x,\xi)\leq \alpha_2\|\xi\|^p, \quad \forall x,\xi\in\R^n
\end{aligned}
\end{equation}
where $\nu:=\tfrac{\partial h}{\partial x}(x) \cdot \xi$,
and such that
\begin{equation} \label{c2}
\begin{aligned}
\frac{\partial V}{\partial x}f+\frac{\partial V}{\partial \xi}\frac{\partial f}{\partial x}\xi\leq -\alpha_3 V
\end{aligned}
\end{equation}
holds.
\end{corollary}
\begin{proof}
The variational system of \eqref{ode} is
\begin{equation}\label{vario}
 \left\{\begin{matrix}
\dot{\xi }=\frac{\partial f}{\partial x}(x(t))\cdot \xi,\\ 
\nu=\frac{\partial h}{\partial x}(x(t))\cdot \xi,
\end{matrix}\right.   
\end{equation} 
If conditions \eqref{c1}-\eqref{c2} are satisfied then using similar arguments as in the proof of Theorem~\ref{theorem1}, the variational system \eqref{vario} is OES. Applying Proposition \ref{proposition1}, we can conclude that the system \eqref{ode} is output contraction. As $y^*$ represents one of admissible output of \eqref{ode}, and \eqref{ode} {is output contractive}, it follows that all the outputs will exponentially converge to $y^*$.
\end{proof}
\begin{example}\label{ex2}
Consider the following time-invariant system 
\begin{equation} \label{ex20}
\left\{\begin{array}{rl}
\dot{x}_{1}& =-3x_2-\sin(x_1+x_2),\\ 
\dot{x}_{2}& =-3x_1+\cos(x_1+x_2),\\
y& =x_1+x_2,
\end{array}\right.
\end{equation}
where $x(t)\in \mathbb{R}^{n}$ is the state vector and $y(t)\in \mathbb{R}^m$ is the output. The left of Figure \ref{fig:3} shows the trajectories of \eqref{ex20} originating from $\sbm{3\\3}$. This plot shows that \eqref{ex20} is unstable. 
The variational system of system \eqref{ex20} is given by
\begin{equation} \label{ex21}
\left\{\begin{array}{rl}
\dot{\xi}_{1} & =-3\xi_{2}-(\xi_{1}+\xi_{2})\cos(x_1+x_2),\\ 
\dot{\xi}_{2}& =-3\xi_{1}+(\xi_{1}+\xi_{2})\sin(x_1+x_2),\\
\nu& =\xi_1+\xi_2.
\end{array}\right.
\end{equation}
By taking $V(x,\xi) = (\xi_1 + \xi_2)^2$ as the Lyapunov function, we can fulfill \eqref{c1} by choosing $\alpha_1 = 1$, $\alpha_2 = 2$, and $p = 2$. The time-derivative of $V(x,\xi)$ satisfies $\dot{V}(x,\xi)=-2\big(3-\cos(x_1+x_2)+\sin(x_1+x_2)\big)(\xi_1+\xi_2)^2$. Consequently, \eqref{c2} can be satisfied by taking $\alpha_3 = 2$. Thus, the system \eqref{ex20} satisfies the hypotheses for OES as given in Corollary \ref{corollary1}. The output associated with the initial condition $\sbm{3\\3}$ is shown in the right plot of Figure~\ref{fig:3}, which shows that it clearly converges to an equilibrium $y^*{\approx 0.246}$.
\begingroup
\begin{figure}[!htb]
 \centering
 \includegraphics[width=3in]{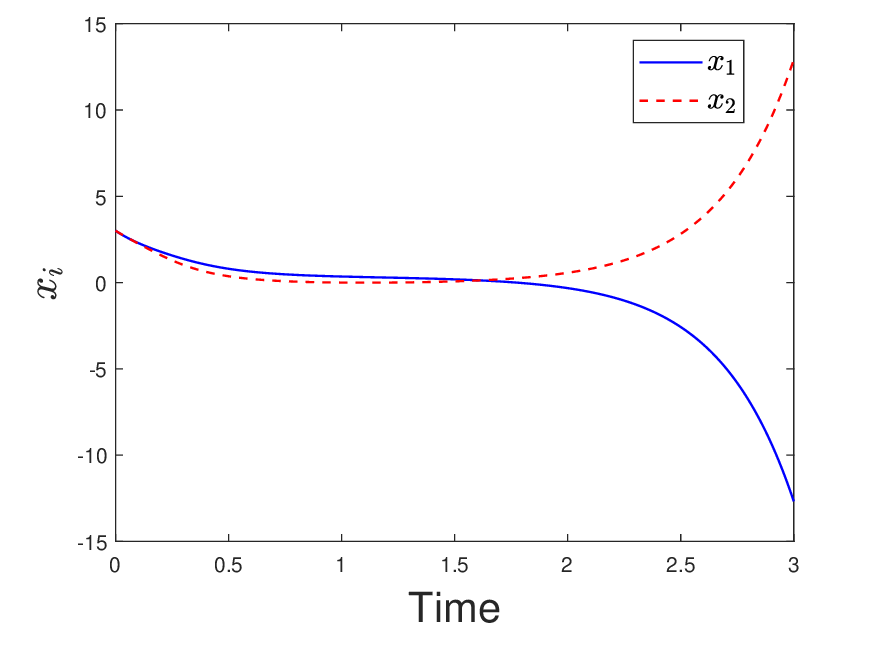}
\includegraphics[width=3in]{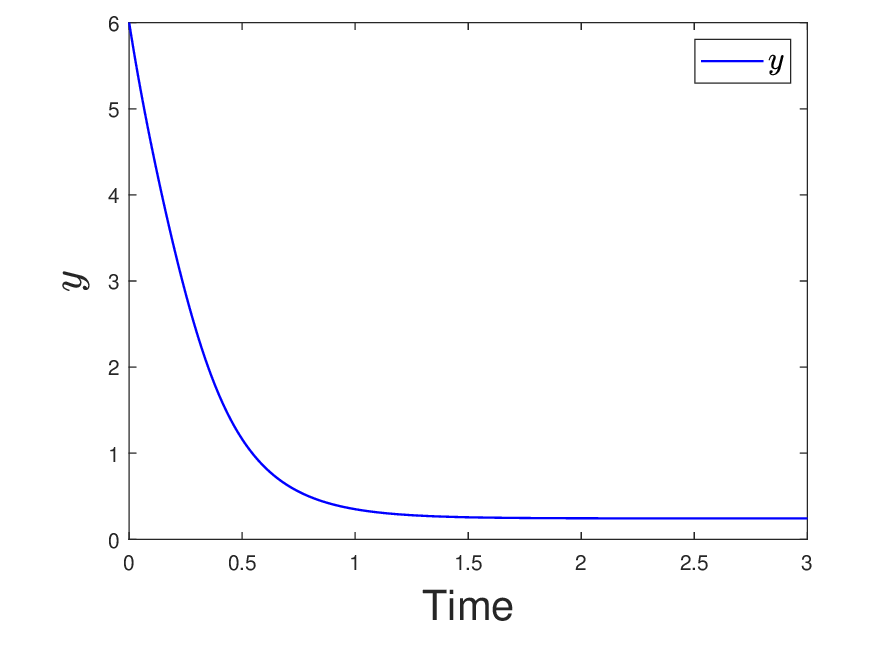}
  \caption{The plot of trajectories and outputs of the system in Example \ref{ex2}
initialized at $\sbm{3\\3}$.} \label{fig:3}
\end{figure}
\endgroup
\end{example}
\section{Conclusion}
This paper introduces the notion of output contraction to characterize the contraction behavior of a nonlinear time-varying system via its output map. We derive a necessary and sufficient condition that establishes a link between the output contraction property of the original system and the output exponential stability (OES) of its variational system. Additionally, we introduce the output contraction Lyapunov condition to guarentee OES of the variational system.

\bibliographystyle{plain}        
\bibliography{autosam}  

\begin{thebibliography}{10}

\bibitem{andrieu2016transverse}
Vincent Andrieu, Bayu Jayawardhana, and Laurent Praly.
\newblock Transverse exponential stability and applications.
\newblock {\em IEEE Transactions on Automatic Control}, 61(11):3396--3411, 2016.

\bibitem{angeli2002lyapunov}
David Angeli.
\newblock A lyapunov approach to incremental stability properties.
\newblock {\em IEEE Transactions on Automatic Control}, 47(3):410--421, 2002.

\bibitem{barabanov2019contraction}
Nikita Barabanov, Romeo Ortega, and Anton Pyrkin.
\newblock On contraction of time-varying port-hamiltonian systems.
\newblock {\em Systems \& Control Letters}, 133:104545, 2019.

\bibitem{davydov2022non}
Alexander Davydov, Saber Jafarpour, and Francesco Bullo.
\newblock Non-euclidean contraction theory for robust nonlinear stability.
\newblock {\em IEEE Transactions on Automatic Control}, 67(12):6667--6681, 2022.

\bibitem{forni2013differential}
Fulvio Forni and Rodolphe Sepulchre.
\newblock A differential lyapunov framework for contraction analysis.
\newblock {\em IEEE transactions on automatic control}, 59(3):614--628, 2013.

\bibitem{gao2008network}
Huijun Gao and Tongwen Chen.
\newblock Network-based $h_{\infty}$ output tracking control.
\newblock {\em IEEE Transactions on Automatic control}, 53(3):655--667, 2008.

\bibitem{isidori1990output}
Alberto Isidori and Christopher~I Byrnes.
\newblock Output regulation of nonlinear systems.
\newblock {\em IEEE transactions on Automatic Control}, 35(2):131--140, 1990.

\bibitem{karafyllis2021relation}
Iasson Karafyllis.
\newblock On the relation of ios-gains and asymptotic gains for linear systems.
\newblock {\em Systems \& Control Letters}, 152:104934, 2021.

\bibitem{lohmiller1998contraction}
Winfried Lohmiller and Jean-Jacques~E Slotine.
\newblock On contraction analysis for non-linear systems.
\newblock {\em Automatica}, 34(6):683--696, 1998.

\bibitem{sontag2014three}
Eduardo~D Sontag, Michael Margaliot, and Tamir Tuller.
\newblock On three generalizations of contraction.
\newblock In {\em 53rd IEEE Conference on Decision and Control}, pages 1539--1544. IEEE, 2014.

\bibitem{wang2005partial}
Wei Wang and Jean-Jacques~E Slotine.
\newblock On partial contraction analysis for coupled nonlinear oscillators.
\newblock {\em Biological cybernetics}, 92(1):38--53, 2005.

\bibitem{wu2023partial}
Chengshuai Wu and Dimos~V Dimarogonas.
\newblock From partial and horizontal contraction to k-contraction.
\newblock {\em IEEE Transactions on Automatic Control}, 2023.

\bibitem{wu2022k}
Chengshuai Wu, Ilya Kanevskiy, and Michael Margaliot.
\newblock k-contraction: Theory and applications.
\newblock {\em Automatica}, 136:110048, 2022.

\end{thebibliography}
\end{document}